\documentclass[conference, a4paper]{IEEEtran}
\IEEEoverridecommandlockouts

\ifCLASSINFOpdf
  \usepackage[pdftex]{graphicx}
  \DeclareGraphicsExtensions{.pdf,.jpeg,.png}
\else
\fi

\usepackage{cite}
\usepackage{amsmath,amssymb,amsfonts,amsthm}
\usepackage{algorithmic}
\usepackage{graphicx}
\usepackage{textcomp}
\usepackage[left=1.4cm, right=1.4cm, top=1.7cm]{geometry}
\usepackage{caption}
\usepackage{bm}
\usepackage{booktabs}
\usepackage{xcolor}

\usepackage{anyfontsize}
\usepackage{placeins}
\usepackage{subcaption, float}
\usepackage{lineno}

\theoremstyle{plain}      

\usepackage{cleveref}
\crefformat{equation}{(#2#1#3)}
\Crefname{algocfline}{Algorithm}{Algorithms}
\Crefname{algocf}{line}{lines}
\Crefname{AlgoLine}{Line}{Lines}
\crefname{AlgoLine}{line}{lines}

\makeatletter
\newcommand*\titleheader[1]{\gdef\@titleheader{#1}}
\AtBeginDocument{%
\let\st@red@title\@title
\def\@title{%
\bgroup\normalfont\normalsize\centering\@titleheader\par\egroup
\vskip0.2em\st@red@title}
}
\makeatother

\makeatletter
\renewcommand{\fnum@figure}{Figure \thefigure}
\makeatother

\title{ Noncooperative Virtual Queue Coordination via Uncertainty-Aware Correlated Equilibria  \\
\thanks{This work was supported by the National Aeronautics and Space Administration ULI Award under grants 80NSSC21M0071 and 80NSSC24M0070, and by the National Science Foundation CAREER award under grants 2336840 and 2211548.}
\vspace{0.5cm}
}

\titleheader{Second US-Europe Air Transportation Research and Development Symposium (ATRDS2026)}

\author{\IEEEauthorblockN{Jaehan Im, David Fridovich-Keil, Ufuk Topcu}
\IEEEauthorblockA{Department of Aerospace Engineering \\
The University of Texas at Austin \\
Austin, Texas, USA \\
jaehan.im@utexas.edu, dfk@utexas.edu, utopcu@utexas.edu}
}

\renewcommand\IEEEkeywordsname{Keywords}
\IEEEaftertitletext{\vspace{-1\baselineskip}}

\begin{document}

\maketitle
\pagestyle{plain}

\newcommand{\activea}{\mathcal{I}}
\newcommand{\aSet}{{\mathcal{I}}}
\newcommand{\xSet}{\bm{x}}
\newcommand{\activef}{\mathcal{F}}
\newcommand{\fSet}{{\mathcal{F}}}

\newcommand{\ones}{\bm 1}
\newcommand{\reals}{{\mbox{\bf R}}}
\newcommand{\integers}{{\mbox{\bf Z}}}
\newcommand{\symm}{{\mbox{\bf S}}}  

\newcommand{\nullspace}{{\mathcal N}}
\newcommand{\range}{{\mathcal R}}
\newcommand{\Rank}{\mathop{\bf Rank}}
\newcommand{\Tr}{\mathop{\bf Tr}}
\newcommand{\diag}{\mathop{\bf diag}}
\newcommand{\card}{\mathop{\bf card}}
\newcommand{\rank}{\mathop{\bf rank}}
\newcommand{\conv}{\mathop{\bf conv}}
\newcommand{\prox}{\bm{prox}}

\newcommand{\Expect}{\mathop{\bf E{}}}
\newcommand{\Prob}{\mathop{\bf Prob}}
\newcommand{\Co}{{\mathop {\bf Co}}} 
\newcommand{\dist}{\mathop{\bf dist{}}}
\newcommand{\argmin}{\mathop{\rm argmin}}
\newcommand{\argmax}{\mathop{\rm argmax}}
\newcommand{\epi}{\mathop{\bf epi}} 
\newcommand{\Vol}{\mathop{\bf vol}}
\newcommand{\dom}{\mathop{\bf dom}} 
\newcommand{\intr}{\mathop{\bf int}}
\newcommand{\sign}{\mathop{\bf sign}}
\newcommand{\norm}[1]{\left\lVert#1\right\rVert}
\newcommand{\mnorm}[1]{{\left\vert\kern-0.25ex\left\vert\kern-0.25ex\left\vert #1 
    \right\vert\kern-0.25ex\right\vert\kern-0.25ex\right\vert}}

\newtheorem{definition}{Definition} 
\newtheorem{theorem}{Theorem}
\newtheorem{lemma}{Lemma}
\newtheorem{corollary}{Corollary}
\newtheorem{remark}{Remark}
\newtheorem{proposition}{Proposition}
\newtheorem{assumption}{Assumption}
\newtheorem{example}{Example}

\newcommand{\cf}{{\it cf.}}
\newcommand{\eg}{{\it e.g.}}
\newcommand{\ie}{{\it i.e.}}
\newcommand{\etc}{{\it etc.}}

\newcommand{\putref}{{\color{red}[r]}}

\newcommand{\ba}[2][]{\todo[color=orange!40,size=\footnotesize,#1]{[BA] #2}}

\newcommand{\fix}[1]{\textcolor{red}{#1}}

\newcommand{\bigO}{\mathcal{O}}

\newcommand{\intSet}{\mathbb{Z}}
\newcommand{\realSet}{\mathbb{R}}
\newcommand{\natSet}{\mathbb{N}}
\newcommand{\zeroSet}{\bm{0}}
\newcommand{\state}{\bm{x}}

\newcommand{\param}{\kappa}

\newcommand{\plSet}{\mathbf{N}}
\newcommand{\chSet}{\mathbf{M}}
\newcommand{\opSet}{\mathcal{O}}

\newcommand{\xVec}{\bm{x}}
\newcommand{\coordFactor}{\mathbf{w}}

\newcommand{\circnum}[1]{%
  \raisebox{.5pt}{\textcircled{\raisebox{-.9pt}{#1}}}%
}

\noindent \begin{abstract}
Collaborative virtual queueing has been proposed as a mechanism to mitigate airport surface congestion while preserving airline autonomy over aircraft-level pushback decisions. A central coordinator can regulate aggregate pushback capacity but cannot directly control which specific aircraft are released, limiting its ability to steer system-level performance. 
We propose a noncooperative coordination mechanism for collaborative virtual queueing based on the correlated equilibrium concept, which enables the coordinator to provide incentive-compatible recommendations on aircraft-level pushback decisions without overriding airline autonomy. 
To account for uncertainty in airlines’ internal cost assessments, we introduce 
chance constraints into the correlated equilibrium formulation. 
This formulation provides explicit probabilistic guarantees on incentive compatibility, allowing the coordinator to adjust the confidence level with which airlines are expected to follow the recommended actions.
We further propose a scalable algorithm for computing chance-constrained correlated equilibria by exploiting a reduced-rank structure. 
Numerical experiments demonstrate that the proposed method scales to realistic traffic levels up to 210 eligible pushbacks per hour, reduces accumulated delay by up to approximately 8.9\% compared to current first-come-first-served schemes, and reveals a trade-off between confidence level, deviation robustness, and achievable cost efficiency.
\end{abstract}

\vspace{0.3cm}


\section{Introduction}

Airport surface congestion is increasingly managed through \textit{collaborative virtual queue} (CVQ) systems \cite{VQ_asdex, VQ_burgain, central_flexibility_VQ}. 
CVQ improves airport surface operation efficiency by replacing physical queues on taxiways with a virtual abstraction thereby avoiding unnecessary engine-on delays.
A key feature of CVQ is the division of authority, in which a central coordinator regulates aggregate gate pushback allowances per airline, whereas aircraft-level pushback decisions are made by individual airlines as illustrated in \Cref{fig:concept}. 

The division of authority in CVQ however, introduces a practical coordination challenge. First, aircraft-level delay costs are airline-specific, time-varying, and private, making them difficult for a central authority to observe or verify. Second, system-level performance is sensitive to which aircraft airlines choose to push back and to the internal decision policies they employ for aircraft selection. Even under identical pushback allowances, uncoordinated airline-specific heuristics can lead to different congestion patterns and surface efficiency outcomes \cite{VQ_asdex, VQ_burgain}. As a result, existing CVQ approaches, which primarily focus on pushback capacity allocation while treating airlines’ internal aircraft selection policies as a black box, provide the coordinator with limited ability to steer aircraft-level decisions toward system-optimal outcomes.

\begin{figure}[t!]
    \centering
    \includegraphics[width=0.9\linewidth]{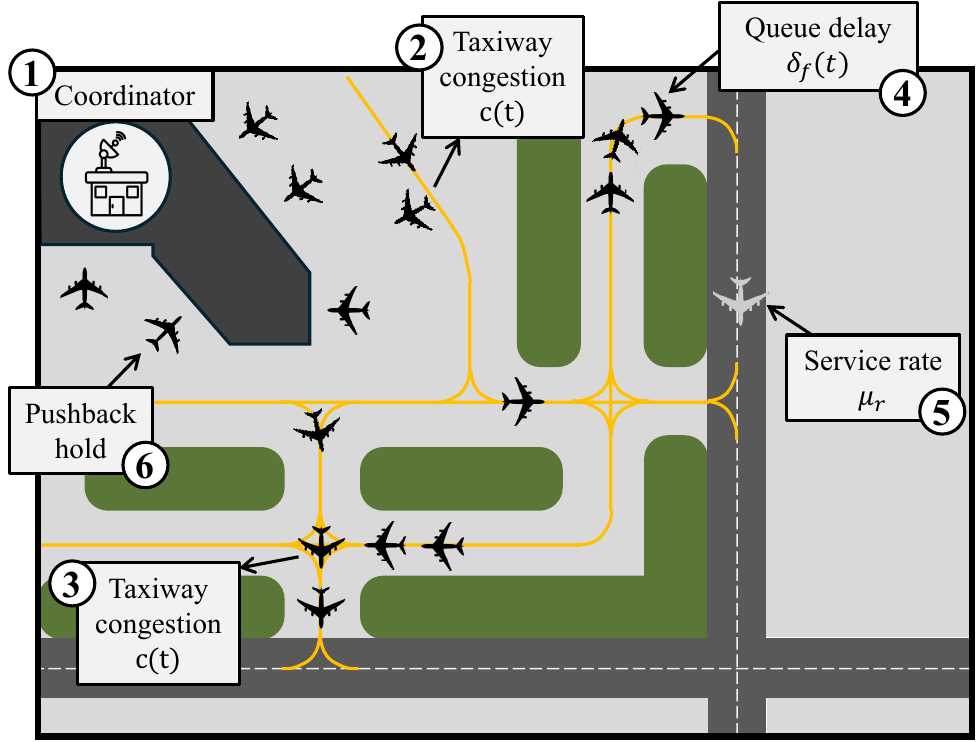}
    \caption{Illustration of the collaborative virtual queue coordination setting.
    A central coordinator regulates aggregate pushback capacity \circnum{1} but does not control aircraft-level decisions directly.
    Released aircraft contribute to taxiway congestion $c(t)$ (\circnum{2}–\circnum{3}) and to runway queue delay $\delta_f(t)$ \circnum{4}, while departures are governed by runway service rates $\mu_r$ \circnum{5}.
    Airlines retain autonomy over pushback decisions \circnum{6}, creating a noncooperative coordination problem.}  
    \label{fig:concept}
\end{figure}

We address this limitation by developing a noncooperative coordination mechanism that preserves airline autonomy while enabling the coordinator to actively steer system-level outcomes.
Rather than overriding airline decisions or imposing centralized control, we leverage the concept of \textit{correlated equilibrium} to provide incentive-aligned recommendations on which aircraft to pushback, in addition to how many.
By design, following the recommended action is rational for each airline given its cost structure, promoting voluntary compliance

Despite its theoretical appeal, correlated equilibrium relies on exact knowledge of agents’ cost functions, which is unrealistic in CVQ settings. To account for uncertainty in airlines’ internal cost assessments, we propose an uncertainty-aware correlated equilibrium framework.
Formally, this uncertainty awareness is captured through a \textit{chance-constrained correlated equilibrium} (CC-CE), 
which ensures that the coordinator's guidance is incentive compatible with at least a user-specified probability.

We further develop a scalable algorithm to compute CC-CE by exploiting a \textit{reduced-rank correlated equilibrium} structure \cite{rrce}.
By reconstructing a tractable subset of correlated equilibria from multiple pure Nash equilibria, the proposed approach significantly reduces the computational burden associated with direct correlated equilibrium computation.
This enables practical deployment of CC-CE in large-scale virtual queue settings.

The contributions of this work are threefold.
First, we reformulate the virtual queue coordination problem to enable aircraft-level influence via a noncooperative coordination mechanism.
By leveraging correlated signaling, the proposed approach allows a CVQ coordinator to steer which specific aircraft are pushed back, while fully preserving airline autonomy.
Second, we introduce a chance-constrained correlated equilibrium formulation that explicitly accounts for uncertainty in airlines’ internal cost assessments.
The proposed formulation provides a probabilistic guarantee on incentive compatibility, allowing the coordinator to quantify the confidence level with which airlines are expected to follow recommended actions under uncertain cost realizations.
Third, we develop a scalable algorithm for computing CC-CE by exploiting a reduced-rank correlated equilibrium structure \cite{rrce}, enabling tractable computation in large-scale CVQ settings.

\section{Related works}

\subsection{Collaborative virtual queueing}

Collaborative virtual queueing (CVQ) has been proposed as a mechanism to mitigate airport surface congestion by regulating gate departure pushbacks while preserving airline autonomy over aircraft-level decisions \cite{VQ_asdex, VQ_burgain, VQ_security}.
It aims to smooth surface traffic and reduce unnecessary taxiing and holding, reflecting the fact that congestion bottlenecks in busy airports typically arise near the runway rather than at the gate \cite{taxiSimul_bottleneck}.

A substantial body of prior work on airport surface management formulates pushback, taxi, and runway sequencing as centralized optimization problems \cite{icn_surface_cent_case, icn_surface_cent_mdi, central_surface_control}, including mixed--integer programs \cite{hamsa_review, central_airport_control, central_gdp, central_departure_schedule}, neural network-based approaches \cite{neural_surface_opt}, and route-optimization schemes \cite{central_route_opt}.
These methods are effective when a single decision maker can enforce aircraft-level actions, but they abstract away airline autonomy and strategic behavior by assuming full centralized control.

Such centralized formulations are misaligned with the authority structure of CVQ.
In CVQ operations, airlines retain control over which specific flights are pushed back within allocated slots as part of collaborative decision making \cite{central_flexibility_VQ}.
Consequently, CVQ cannot directly influence airlines’ internal flight-selection policies, and system-level performance may still degrade due to airline-specific decision heuristics \cite{VQ_asdex, VQ_burgain}.

In a similar spirit, several noncooperative and game-theoretic approaches have been explored in related air traffic and mobility contexts, although not directly in the virtual queueing setting \cite{game_aam, auction_aam, j_game_decent, j_nash, j_nego_over}.
These studies demonstrate how taxation \cite{j_nego_over}, bilevel formulations \cite{game_aam}, or auction mechanisms \cite{auction_aam, j_nego_over} can guide stakeholder behavior without assuming centralized enforcement.
However, they typically do not model the interactive decision-making among competing airlines \cite{game_aam} and often rely on strong assumptions such as truthful reporting \cite{auction_aam}, or fixed cost models without explicitly accounting for uncertainty in stakeholders’ internal cost structures \cite{rrce, game_aam, auction_aam, j_game_decent, j_nego_over}.

\subsection{Correlated equilibria under uncertainty}

The correlated equilibrium concept, introduced by Aumann as a generalization of the Nash equilibrium \cite{aumann}, has been widely studied as a mechanism for coordinating noncooperative agents, as it allows a coordinator to influence joint behavior while respecting individual's ability to make decisions that optimize their own objectives \cite{rrce, ce_effective_traffic, ce_effective_experimental}.
Existing work on correlated equilibria under uncertainty can be broadly grouped into several branches \cite{ce1,ce2,ce3,ce4,ce6,ce7}.

One line of research studies Bayes correlated equilibria, which model payoff uncertainty through states, types, and information structures and define rationality in terms of conditional expected utilities \cite{ce1,ce3}. In these formulations, uncertainty arises from players having incomplete information about payoff-relevant states or types; conditional on a given state, payoffs are deterministic. As a result, incentive compatibility is defined in expectation with respect to players’ beliefs, rather than through explicit probabilistic guarantees on realized deviation gains \cite{ce1,ce3,ce6,ce7}.

Other class of approaches models uncertainty directly in the agents' cost function through robust or distributionally robust formulations \cite{ce2,ce4,ce6,ce8,ce9}.
By enforcing worst-case guarantees over all admissible realizations or distributions of the payoff uncertainty within an ambiguity-set, these methods offer strong robustness but often lead to nonconvex optimization problems \cite{ce8,ce9}.
As a result, they cannot fully exploit the linear and convex structure of the classical correlated equilibrium problem and may be overly conservative in applications where uncertainty arises from stochastic fluctuations around nominal costs.

Finally, bounded-rationality models such as quantal response equilibria and their correlated variants replace exact best responses with noisy, payoff-responsive choice rules \cite{qre1,qre2}. Equilibrium is defined as a fixed point of these stochastic response mappings. However, these models do not impose explicit probabilistic guarantees on profitable deviations; randomness is embedded in the response mechanism rather than regulated through deviation-based reliability constraints.

\section{Collaborative virtual queue problem}
\label{sec:cvq}
We formulate collaborative virtual queue (CVQ) coordination as a noncooperative game among airlines, with a ground controller acting as a coordinator. At each decision epoch, a coordinator determines the aggregate pushback capacity, while airlines retain autonomy over aircraft-level pushback decisions.
The coordinator does not enforce aircraft-level actions, but instead provides recommendations with the objective of minimizing system-level surface delay by influencing which aircraft are pushed back at each decision epoch.

\subsection{Virtual queue coordination problem}

Time is discretized into decision \emph{epochs} indexed by
$t \in \intSet_{\geq 0}$.
Let $R \in \intSet_{> 0}$ denote the number of departure runways at the airport and $\mathcal F(t)$ denote the set of aircraft eligible for pushback at epoch $t$.
At each epoch $t$, the coordinator observes the surface \emph{state}
$s(t) := \big(q(t), \mathcal F(t)\big) \in \mathcal{S}$,
where $\mathcal S$ denotes the state space, 
and $q(t) \in \intSet_{\ge 0}^{R}$ denotes the vector of runway queue lengths whose $r$-th entry
$q_r(t)$ represents the number of aircraft physically waiting for departure at runway $r \in \{1,\dots,R\}$.

Let $\activea(t)$ denote the set of airlines that have at least one
eligible aircraft at epoch $t$.
For each airline $i \in \activea(t)$, let
$\mathcal F_i(t) \subseteq \mathcal F(t)$
denote the subset of eligible aircraft operated by airline $i$.

\subsubsection{Airline actions}
For each active airline $i \in \activea(t)$, let $\mathcal X_i(t)$ denote the set of feasible aircraft-level pushback actions available to airline $i$ at epoch $t$.
Airline $i$ selects an action $x_i(t)$ defined as:
\begin{equation}
x_i(t) \in \mathcal X_i(t) := 2^{\mathcal F_i(t)} .
\end{equation}
The resulting joint action is denoted by
\begin{equation}
x(t) := \big(x_i(t)\big)_{i \in \activea(t)} \in \mathcal X(t),
\end{equation}
where $\mathcal X(t) := \prod_i \mathcal X_i(t)$.
This action representation directly determines which aircraft enter the surface queues at epoch $t$.

\subsubsection{Eligible aircraft set dynamics}
Let $\mathcal P(t)$ denote the set of aircraft scheduled to depart at epoch $t$ according to
the published flight schedule.
Each aircraft $f \in \mathcal P(t)$ is associated with a scheduled departure time $\tau_f$, a designated departure runway $r(f) \in \{1,\dots,R\}$, and an aircraft class (e.g., small, medium, or heavy).

A set of aircraft that were in an eligible set but not pushed back in the previous epoch $t-1$ remain eligible in subsequent epoch $t$. Accordingly, the eligible aircraft set evolves as follows:
\begin{equation}
\mathcal F(t) =
\mathcal P(t)
\;\cup\;
\bigl\{ f \in \mathcal F(t-1) \;\big|\; f \notin \mathcal{D}(t-1) \bigr\},
\label{eq:eligible_update}
\end{equation}
where $\mathcal{D}(t) := \bigcup_{i \in \activea(t)} x_i(t)$ denotes the aircraft that are physically pushed back in epoch $t$.

\subsubsection{Runway queue dynamics}
Let runway service rate $\mu_r \in \intSet_{\geq 0}$ denote the number of aircraft that departs from runway $r$ per epoch, and let
\begin{equation} \label{eq:pushbacknumber}
a_r(t) := \sum_{f \in \mathcal D(t)} \ones \{r(f)=r\},
\end{equation}
where $\ones$ is an indicator function.
\Cref{eq:pushbacknumber} denote the number of aircraft pushed back to runway $r$ at epoch $t$.
The runway queue evolves according to
\begin{equation}
q_r(t+1)
=
\max\bigl\{ 0,\, q_r(t) + a_r(t) - \mu_r\bigr\},
\quad \forall r.
\label{eq:queue_update}
\end{equation}

\subsubsection{Surface state evolution}
Combining \cref{eq:eligible_update} and \cref{eq:queue_update}, the surface state evolves according to
\begin{equation}
s(t+1) = g\bigl(s(t), \xSet(t)\bigr),
\end{equation}
where $g$ captures both runway service dynamics and the evolution of the eligible aircraft set.

\subsubsection{Airline costs}

Airlines are modeled as noncooperative agents minimizing individual operational costs.
At epoch $t$, airline $i$ incurs a cost
\begin{equation} \label{eq:airCost}
J_i^t : \mathcal X(t) \times \mathcal S \;\to\; \realSet.
\end{equation}

The cost $J_i^t(\xSet(t); s(t))$ captures aircraft class-weighted delay and congestion effects induced by the joint pushback decisions and the resulting surface dynamics.
For simplicity, we assume that airlines myopically optimize \Cref{eq:airCost} at the current $t$ and do not consider future costs.

\subsubsection{Coordinator objective}

The coordinator evaluates joint actions using a system-level objective
\begin{equation} \label{eq:coordCost}
J_{\mathrm{coord}}^t : \mathcal X(t) \times \mathcal S \;\to\; \realSet,
\end{equation}
which aggregates aircraft-level delays over the epoch following the pushback decisions. 
The mecahnism by which the coordinator induces an action introduced in \Cref{sec:ce}.

\subsection{Incorporating airline cost uncertainty}

The coordinator does not observe airlines' true cost functions $J_i^t$.
Instead, it maintains nominal cost models $\bar J_i^t$, which may differ from the airlines' realized operational costs.

We model this mismatch in terms of unilateral deviations.
Specifically, for any fixed joint action profile $(x_i,x_{-i})$ and any unilateral alternative $x_i'$, 
the true deviation cost difference is assumed to satisfy
\begin{equation}
\begin{aligned}
 J_i&(x_i,x_{-i}) -  J_i(x_i',x_{-i})\\
&=\bar J_i(x_i,x_{-i}) - \bar J_i(x_i',x_{-i})+\eta_i,
\end{aligned}
\label{eq:noise_model}
\end{equation}
where random variable $\eta_i \sim \nu_i$ represents airline-level uncertainty in deviation incentives and is drawn once per airline at each epoch.
Importantly, $\eta_i$ is common across all deviation comparisons for airline $i$ within the same epoch.

\medskip

In every epoch, each airline chooses an action $x_i(t)$ so as to minimize \cref{eq:airCost}, while the coordinator's objective is to induce a joint action that minimizes \cref{eq:coordCost}. However, the coordinator cannot enforce aircraft-level actions or demand truthful disclosure of airline cost models. 
These issues motivate the proposed coordination mechanism, discussed below.

\section{Coordination via correlated equilibrium \\ in virtual queue problem}
\label{sec:ce}

The correlated equilibrium concept provides a coordination mechanism for noncooperative agents without requiring centralized enforcement.
Rather than prescribing a deterministic joint action, a coordinator samples an action profile from a joint distribution and privately recommends that each agent follows the corresponding action.
Following the recommendation is voluntary; nevertheless, rational agents find it optimal to comply when the distribution satisfies the correlated equilibrium conditions.

\subsection{Correlated equilibrium}
Formally, consider the game induced at a fixed surface state $s(t)$.
Recall that $\mathcal X := \prod_{i\in\mathcal I(t)} \mathcal X_i(t)$ denotes the joint action space, and
let $z \in \Delta(\mathcal X)$ be a probability distribution over joint actions.

\begin{definition}[Correlated Equilibrium] \label{def:ce}
A distribution $z \in \Delta(\mathcal X)$ is a correlated equilibrium if the following holds:
\begin{equation}
\mathbb E_{x_{-i} \sim z(\cdot \mid x_i)}
\!\left[
J_i(x_i, x_{-i}; s(t)) - J_i(x_i', x_{-i}; s(t))
\right]
\le 0,
\label{eq:ce_def}
\end{equation}
for every airline $i$, every recommended action $x_i \in \mathcal X_i(t)$, and every unilateral deviation $x_i' \in \mathcal X_i(t)$.
\end{definition}

The following lemma describes a basic geometric property of correlated equilibria.

\begin{lemma}[Convexity of correlated equilibria] \label{lem:ceiscvx}
The set of correlated equilibria is a convex subset of the probability simplex $\Delta(\mathcal X)$.
\end{lemma}

\begin{proof}
The constraints in \Cref{eq:ce_def} are linear in $z$, and $\Delta(\mathcal X)$ is convex.
Hence, the set of correlated equilibria is convex.
\end{proof}

Under a correlated equilibrium $z$, the coordinator samples a joint action from $z$ and privately recommends each component to the corresponding airline. 
Given $z$, each airline minimizes its expected cost conditional on the recommendation. 
By \Cref{def:ce}, following the recommendation is incentive-compatible. 
Thus, the coordinator neither modifies airlines’ objective functions nor enforces compliance; coordination arises naturally from incentive alignment. 
This makes correlated equilibrium well suited for collaborative virtual queues (CVQ), where airline autonomy must be preserved.

\begin{example}[Coordination via correlation]
Consider a scenario where two drivers approach an intersection and choose either 
\emph{Go} (G) or \emph{Stop} (S).
Their costs are given by
\small{
\begin{equation}
\begin{array}{c|cc}
 & G & S \\
\hline
G & (5,5) & (-1,1) \\
S & (1,-1) & (1,1)
\end{array}
\end{equation}
}
where rows correspond to Driver~1 and columns to Driver~2.
If both go, they incur a large penalty (5,5).
If both stop, both incur a mild delay (1,1).
If exactly one goes, the driver who goes benefits ($-1$) 
while the other incurs a mild delay ($1$).
Consider a correlation device that selects $(G,S)$ and $(S,G)$ with probability $\frac{1}{2}$.
If a driver is recommended to go, following the recommendation yields cost $-1$,
whereas deviating leads to $(S,S)$ with cost $1$.
If recommended to stop, following yields cost $1$,
whereas deviating leads to $(G,G)$ with cost $5$.
Thus, conditioned on the received signal, deviation strictly increases cost.
The correlation device therefore avoids collision or inefficient mutual stopping.
\end{example}

\subsection{Chance-constrained correlated equilibrium} \label{sec:ccce}
In practice, the coordinator does not have exact knowledge of airlines' cost functions.
Deviation incentives are therefore uncertain and must be evaluated probabilistically.
To account for this, we extend the correlated equilibrium concept.

For airline $i$, given an action profile $(x_i,x_{-i})$ and a deviation $x_i'$, 
define the deviation cost difference under the true model as
\begin{equation}
\Delta J_i(x_i,x_i',x_{-i})
:= J_i(x_i,x_{-i}) - J_i(x_i',x_{-i}),
\end{equation}
and under the nominal model as
\begin{equation}
\bar{\Delta} J_i(x_i,x_i',x_{-i})
:= \bar J_i(x_i,x_{-i}) - \bar J_i(x_i',x_{-i}).
\end{equation}
Recall that the realized deviation cost difference is subject to additive stochastic uncertainty,
\begin{equation}
\Delta J_i(x_i,x_i',x_{-i})
=
\bar{\Delta} J_i(x_i,x_i',x_{-i})
+
\eta_i,
\label{eq:noise}
\end{equation}
where $\eta_i \sim \nu_i$ captures modeling error and operational variability in airline $i$'s cost.
The disturbance $\eta_i$ represents an airline-level common error and is assumed to be identical across all realizations of $x_{-i}$, rather than being resampled for each action profile.
For numerical evaluation in \Cref{sec:experiment}, we consider the special case $\nu_i = N(0, \sigma_i^2)$, where $\sigma_i$ is measured in the same cost unites as \Cref{eq:noise}.

This leads to the following definition of a chance-constrained correlated equilibrium.
\begin{definition}[Chance-constrained correlated equilibrium]
A distribution $z \in \Delta(\mathcal X)$ is a chance-constrained correlated equilibrium (CC-CE)
with confidence level $\alpha \in [0,1]$ if,
\begin{equation}
\mathbb P_{\nu_i}\!\left(
\mathbb E_{x_{-i}\sim z(\cdot \mid x_i)}
\bigl[
\bar \Delta J_i(x_i,x_i',x_{-i})
\bigr]
\le 0
\right)
\ge \alpha,
\label{eq:ccce_prob}
\end{equation}
for all airlines $i$, all recommended actions $x_i \in \mathcal X_i(t)$, and all
deviations $x_i' \in \mathcal X_i(t) \setminus x_i$,
\end{definition}

More generally, suppose that the deviation distribution $\nu_i$ admits a cumulative distribution function $\Phi_{\nu_i}: \realSet \to [0,1]$.
Then the chance constraint~\cref{eq:ccce_prob} can be expressed as:
\begin{equation}
\mathbb E_{x_{-i}\sim z(\cdot \mid x_i)}
[\bar{\Delta} J_i(x_i,x_i',x_{-i})]
+
\Phi_{\nu_i}^{-1}(\alpha)
\le 0 ,
\label{eq:ccce_det}
\end{equation}
where $\Phi_{\nu_i}^{-1}$ denotes the inverse cumulative distribution function of distribution $\nu_i$.

An important property of CC-CE is that it preserves the convexity of the classical correlated equilibrium set as stated in \cref{lem:ceiscvx}.

\begin{theorem}[Convexity of CC-CE] \label{thm:convex_ccce}
The feasible set of chance-constrained correlated equilibria is a convex polytope.
\begin{proof}
Each chance constraint reduces to an affine inequality in the joint distribution $z$ as shown in \cref{eq:ccce_det}.
Since $\Delta(\mathcal X)$ is convex, the CC-CE feasible set is an intersection of finitely many affine half-spaces and is therefore a convex polytope.
\end{proof}
\end{theorem}
The confidence level $\alpha$ explicitly quantifies the tradeoff between robustness and feasibility:
larger $\alpha$ yields stronger incentive guarantees but shrinks the feasible set.

\subsection{Scalable computation via reduced-rank correlated equilibria}
Direct computation of correlated equilibrium is known for intractability to the number of the joint actions in the game. This problem holds in CC-CE computation and in virtual queue coordination problem as well since the joint action space grows exponentially with the number of airlines and aircraft.
To address this challenge, we build on the reduced-rank correlated equilibrium (RRCE) framework~\cite{rrce}.

RRCE approximates the correlated equilibrium polytope by the convex hull of a finite set of Nash equilibria, exploiting the fact that Nash equilibria can be computed without enumerating all joint actions and are therefore computationally less demanding than solving for a correlated equilibrium directly.
This idea relies on two classical results: every Nash equilibrium is a correlated equilibrium \cite{netoce}, and the set of correlated equilibria is convex (c.f., \cref{lem:ceiscvx})). 
We extend these ideas to computing the reduced-rank CC-CEs.

\begin{definition}[Chance-constrained pure Nash equilibrium]
A pure strategy profile $x \in \mathcal X$ is a chance-constrained pure Nash equilibrium (CC-PNE)
with confidence level $\alpha$ if,
\begin{equation} \label{eq:ccpne_prob}
\mathbb P_{\nu_i}\!\left(
\Delta J_i(x_i, x_i', x_{-i}) \le 0
\right)
\ge \alpha ,
\end{equation}
for every airline $i$ and every unilateral deviation $x_i' \in \mathcal X \setminus x_i$.
\end{definition}

The CC-PNE has a useful relationship to CC-CE as shown in the following lemma.
\begin{lemma} \label{lem:pnetoccce}
Every CC-PNE induces a CC-CE distribution.
\begin{proof}
Let $x \in \mathcal X$ be a CC-PNE with confidence level $\alpha$.
By the definition of pure Nash equilibrium, the joint action distribution $z$ concentrates all mass on $x$, i.e., $z(x)=1$ and $z(\tilde x)=0$ for all $\tilde x \neq x$. Thus, the following holds:
\begin{equation} \label{eq:lem1_link}
\begin{aligned}
    \mathbb P&\!\left(
    \mathbb E_{\tilde x_{-i}\sim z(\cdot \mid x_i)}
    \bigl[\Delta J_i(x_i,x_i',\tilde x_{-i})\bigr]
    \le 0
    \right) \\
    &= 
    \mathbb P\!\left(
    \Delta J_i(x_i,x_i',x_{-i}) \le 0
    \right).
\end{aligned}
\end{equation}
\noindent
Since $x$ is a CC-PNE, we can combine \cref{eq:ccpne_prob} with \cref{eq:lem1_link} yielding:
\begin{equation}
\mathbb P\!\left(\mathbb E_{\tilde x_{-i}\sim z(\cdot \mid x_i)}
\bigl[\Delta J_i(x_i,x_i',\tilde x_{-i})\bigr]\le 0\right) \geq \alpha.
\end{equation}
This completes the proof.
\end{proof}

\end{lemma}
\noindent
This observation extends naturally to convex combinations of CC-PNE.

\begin{theorem} \label{thm:ccpneisccce}
The convex hull of CC-PNE distributions is a subset of the CC-CE feasible set.
\begin{proof}
Every CC-PNE induces a CC-CE distribution by \Cref{lem:pnetoccce}.
Since the CC-CE feasible set is convex by \Cref{thm:convex_ccce},
any convex combination of CC-PNE distributions also lies in the CC-CE set.
\end{proof}
\end{theorem}

This result justifies extending the RRCE paradigm to the uncertainty-aware setting.
By computing a finite collection of CC-PNE and constructing their convex hull,
the coordinator can generate CC-CE distributions without solving the full CC-CE program. This approach significantly improves scalability in large instances \cite{rrce}.

\subsection{Coordination via correlated equilibrium selection}

We introduced the notion of correlated equilibrium and its chance-constrained extension to incorporate cost uncertainty, along with a reduced-rank construction for scalable computation. 
We now formalize the noncooperative coordination mechanism based on these equilibrium concepts.

At each decision epoch $t$, the coordinator does not enforce a deterministic joint action.
Instead, it computes an optimal distribution $z \in \Delta(\mathcal X(t))$, samples a joint action $x \sim z$, and privately recommends to each airline its corresponding component $x_i$.
Compliance is voluntary as those recommendations are incentive-compatible to the airlines.
Therefore, noncooperative coordination via this correlated equilibrium mechanism reduces to an equilibrium selection problem over the feasible CC-CE set.

\subsubsection{Full CC-CE coordination problem}
Recall that $J_{\mathrm{coord}}^t(x; s(t))$ denotes the system-level objective evaluated at joint action $x \in \mathcal X(t)$. 
The expected system cost under distribution $z$ is
\begin{equation} \label{eq:coordObjFunc}
\mathbb E_{x \sim z}
\left[
J_{\mathrm{coord}}^t(x; s(t))
\right]
=
\sum_{x \in \mathcal X(t)} z(x)\,
J_{\mathrm{coord}}^t(x; s(t)).
\end{equation}

Then, the coordination problem based on CC-CE is formulated as:
\begin{equation}
\begin{aligned}
\min_{z \in \Delta(\mathcal X(t))} \quad
& \sum_{x \in \mathcal X(t)} z(x)\,
J_{\mathrm{coord}}^t(x; s(t)) \\
\text{s.t.} \quad
& z \text{ satisfies \Cref{eq:ccce_prob}}.
\end{aligned}
\label{eq:ccce_opt}
\end{equation}

Since the CC-CE constraints are affine in $z$ and the simplex constraints are linear,
\cref{eq:ccce_opt} is a linear program.
This formulation provides the optimal incentive-compatible coordination policy with $\alpha \in [0, 1]$ confidence.

\subsubsection{Reduced-rank CC-CE coordination}

To improve scalability, we restrict the feasible set to convex combinations of CC-PNE.
Let $\{x^{(k)}\}_{k=1}^m$ denote a finite set of $m$ CC-PNEs.
Define the distribution
\begin{equation} \label{eq:xtoz}
z(x) = \sum_{k=1}^m \lambda_k \mathbf{1}\{x = x^{(k)}\},
\quad
\lambda \in \Delta_m,
\end{equation}
where $\Delta_m$ is the $m$-dimensional simplex and $\ones$ is an indicator function. 
Substituting \Cref{eq:xtoz} to \Cref{eq:coordObjFunc} yields
\begin{equation}
\mathbb E_{x \sim z}
\left[
J_{\mathrm{coord}}^t(x; s(t))
\right]
=
\sum_{k=1}^m \lambda_k
J_{\mathrm{coord}}^t(x^{(k)}; s(t)).
\end{equation}

Then, the reduced coordination problem becomes
\begin{equation}
\begin{aligned}
\min_{\lambda \in \Delta_m} \quad
& \sum_{k=1}^m \lambda_k\,
J_{\mathrm{coord}}^t(x^{(k)}; s(t)).
\end{aligned}
\label{eq:rrccce_opt}
\end{equation}

By \Cref{thm:ccpneisccce}, every feasible solution of
\cref{eq:rrccce_opt} induces a valid CC-CE distribution.
This reduced formulation significantly lowers computational complexity,
as it avoids enumerating the full joint action space and instead operates over
a small set of CC-PNE.

\medskip

In summary, noncooperative coordination in CVQ is achieved by solving an equilibrium selection problem:
either the full CC-CE program~\cref{eq:ccce_opt} or its reduced-rank counterpart~\cref{eq:rrccce_opt}.
The former provides optimal coordination under uncertainty,
while the latter offers a scalable approximation suitable for large-scale airport operations.

\section{Numerical experiment} \label{sec:experiment}
\subsection{Scenario description}
We consider a collaborative virtual queue (CVQ) surface-operations scenario with each epoch corresponding to four minutes of real time.
At each epoch, the coordinator observes the current runway queue lengths $q(t)$ and the set of eligible aircraft $\fSet(t)$.
Each aircraft $f\in\fSet(t)$ is associated with a designated runway $r(f)$ and an aircraft class (i.e., small/medium/heavy). Airlines select subsets of their eligible aircraft for pushback, inducing runway inflows $a_r(t)$.
The runway queues evolve according to the dynamics introduced in \Cref{sec:cvq}, with service rates $\mu_r$ aircraft per epoch. Airlines act according to their realized cost \cref{eq:noise}; in particular, deviation incentives are subject to additive Gaussian perturbations $\nu_i \sim \mathcal{N}(0,\sigma_i^2)$ that are unobserved by the coordinator. The overall CVQ coordination loop is illustrated in \Cref{fig:concept}.

\subsubsection{Costs and objectives}

We evaluate costs over the eligible flight set $\fSet(t)$ under the realized joint pushback decision $x(t)$.
Let $\tau_f$ denote the scheduled departure time of flight $f$ in minutes.
Since each epoch represents four minutes, the accumulated lateness (in minutes) at epoch $t$ is defined as
\begin{equation}
\ell_f(t):=\max\{0,\, 4t-\tau_f\},
\end{equation}
which captures realized schedule delay.
We define the one-epoch waiting penalty
\begin{equation}
w_f(t):=
\begin{cases}
0, & f\in \mathcal D(t),\\
4, & f\notin \mathcal D(t),
\end{cases}
\end{equation}
where $\mathcal D(t)$ is the set of pushed flights; thus, a flight that is not released incurs an additional four-minute delay.

Let $\tilde q_r(t+1, x(t))$ denote the predicted next-epoch queue length at runway $r$ induced by $x(t)$: $\tilde q_r(t+1, x(t)) = q_r(t) + a_r(t)$, where $a_r(t)$ is the number of flights in $x(t)$ assigned to runway $r$ as defined in \Cref{eq:pushbacknumber}.
The runway queue-induced delay for flight $f$ is
\begin{equation}
\delta_f(t)
=
\frac{\tilde q_{r(f)}(t+1, x(t))}{\mu_{r(f)}} \cdot 4,
\end{equation}
which approximates the expected waiting time (in minutes) required for proceding aircraft in the departure queue to be served.
We also include a taxiway congestion penalty
\begin{equation}
c(t)=\bigl(\max\{0,\, |\mathcal D(t)|-4\}\bigr)^2,
\end{equation}
which reflects that the taxiway can accommodate up to four simultaneously released aircraft without congestion effects; additional releases beyond this threshold incur a quadratic penalty.
The per-flight delay is
{\small
\begin{equation}
d_f(t)=
\begin{cases}
(\ell_f(t)+w_f(t))^2 + \delta_f(t) + c(t), & \ell_f(t)+w_f(t)>10,\\
\ell_f(t) + w_f(t) + \delta_f(t) + c(t), & \text{otherwise}.
\end{cases}
\end{equation}
}

\noindent When the accumulated pushback delay $\ell_f(t)+w_f(t)$ exceeds 10 minutes,
the airline incurs a quadratic penalty, reflecting increased airline/passenger dissatisfaction. This delay model reflects a trade-off between congestion from schedule delay from excessive pushbacks, and is consistent with commonly used delay models in departure-management studies \cite{quadratic}.

Airline $i$ minimizes the class-weighted delay
\begin{equation} \label{eq:airobj}
J_i^t(x(t);s(t))=\sum_{f\in\fSet_i(t)} \omega_{c(f)}\, d_f(t),
\end{equation}
where $\omega_{c(f)} \in \realSet_{\geq 0}$ is a weight determined by the aircraft class of flight $f$.
In the experiments, we use
$
[\omega_{\text{heavy}}, \omega_{\text{medium}}, \omega_{\text{small}}]
= [1.2, 1.0, 0.75]
$.
The coordinator evaluates the unweighted total delay
\begin{equation} \label{eq:coordObj}
J_{\mathrm{coord}}^t(x(t);s(t))=\sum_{f\in\fSet(t)} d_f(t).
\end{equation}
Both airlines and the coordinator seek to reduce delay, but airlines apply class-dependent weights whereas the coordinator evaluates unweighted total delay. This reflects a standard distinction between airline objectives (e.g., fleet-dependent valuations) and an airport/coordinator objective that prioritizes aggregate throughput and surface delay \cite{VQ_asdex,VQ_burgain}.

\subsection{Evaluation metrics}
We evaluate solver performance using three metrics:

\paragraph{\textbf{Delay cost}}
We report the realized delay cost computed according to the cost model defined in \cref{eq:coordObj}.
This metric captures both schedule delay and congestion-induced delay.

\paragraph{\textbf{Computation time}}
We measure wall-clock computation time required to compute the equilibrium at each epoch.

\paragraph{\textbf{Deviation rate}}
We report the deviation rate, defined as the fraction of trials in which at least one airline deviates from the recommended action.

\subsection{Baseline algorithms}
We compare five coordination mechanisms:

\paragraph{\texttt{\texttt{CENT}}}
A fully centralized benchmark assuming complete control authority over all aircraft decisions. This provides an idealized lower bound on achievable cost.

\paragraph{\texttt{\texttt{FCFS}}}
A decentralized first-come-first-served mechanism reflecting current operational practice \cite{VQ_asdex,VQ_burgain}. 
The coordinator solves the same aggregate capacity optimization problem as in \texttt{CENT}, but does not select specific aircraft. Each airline then independently selects the aircraft with the earliest scheduled departure.

\paragraph{\texttt{\texttt{Full-CCCE}}}
A full correlated equilibrium solver that explicitly accounts for cost uncertainty using chance-constrained correlated equilibrium (CC-CE) formulation \cref{eq:ccce_opt}.

\paragraph{\texttt{\texttt{RR-Nominal}}}
A reduced-rank correlated equilibrium (RRCE)-based solver that solves 
\cref{eq:rrccce_opt} under nominal costs, without accounting for cost uncertainty.

\paragraph{\texttt{\texttt{RR-CCCE}}}
A RRCE-based solver that incorporates cost uncertainty into the equilibrium constraints.

\medskip

Table~\ref{tab:baseline} summarizes the key differences among the compared coordination mechanisms.
\begin{table}[h]
\centering
\caption{Comparison of coordination mechanisms}
\label{tab:baseline}
\begin{tabular}{lcc}
\toprule
Method & CE formulation & Uncertainty-aware \\
\midrule
\texttt{CENT} & -- & -- \\
\texttt{FCFS} & -- & -- \\
\texttt{Full-CCCE} & Full formulation \cref{eq:ccce_opt} & Yes \\
\texttt{RR-Nominal} & Reduced formulation \cref{eq:rrccce_opt} & No \\
\texttt{RR-CCCE} & Reduced formulation \cref{eq:rrccce_opt} & Yes \\
\bottomrule
\end{tabular}
\end{table}

\begin{remark}
For the RRCE-based methods, we compute the set of pure Nash equilibria by enumerating all joint actions and checking the equilibrium condition.
While exhaustive enumeration is utilized, the considered single-epoch instances remain of moderate size, while reflecting realistic busy-airport conditions \cite{example_atlanta}, making this approach computationally tractable.
\end{remark}

\subsection{Experimental plan}

We evaluate the proposed coordination mechanism at two levels.
First, we conduct a single-epoch Monte Carlo experiment.
This setting enables a controlled comparison of equilibrium computation and resulting costs across solvers. 
Second, we perform a 16-epoch (one-hour simulation time) experiment in which the single-epoch mechanism is applied sequentially. This experiment assesses cumulative delay behavior and demonstrates operational applicability.

\subsubsection{Experiment 1 (Single-epoch simulation)}
In the single-epoch experiment, each trial randomly generates a busy-airport epoch instance. We fix the number of runways ($R=2$), service rates ($[\mu_1,\mu_2]=[2,2]$), and initial queue lengths $q(0)=[3,4]$.
The experimental variables are:
(i) the eligible aircraft count $|\fSet(t)|$,
(ii) the confidence level $\alpha$, and
(iii) the deviation-noise level $\sigma$ introduced in \Cref{sec:ccce}. 
We evaluate $|\fSet(t)| \in \{6,7,\dots,14\}$,
$\alpha \in [0.30, 0.99]$, and
$\sigma \in [0,45]$ (recall that $\sigma$ is measured in the same cost units as \Cref{eq:noise}).
For each eligible flight, scheduled lateness is independently drawn as 
$L_f \sim \mathcal{N}(0, 10^2)$ and truncated via $L_f \gets \max\{0, L_f\}$, so that negative realizations correspond to on-time departures.

The aircraft count $|\fSet(t)|=14$ per epoch corresponds to approximately 210 eligible pushbacks per hour.
For reference, Hartsfield-Jackson Atlanta International Airport—one of the busiest airports worldwide—processes roughly 96 departures per hour \cite{example_atlanta}.
Thus, the tested range spans realistic and congested regimes.

\subsubsection{Experiment 2 (Multi-epoch rolling simulation)}

\begin{figure}[t!]
    \centering
    \includegraphics[width=0.85\linewidth]{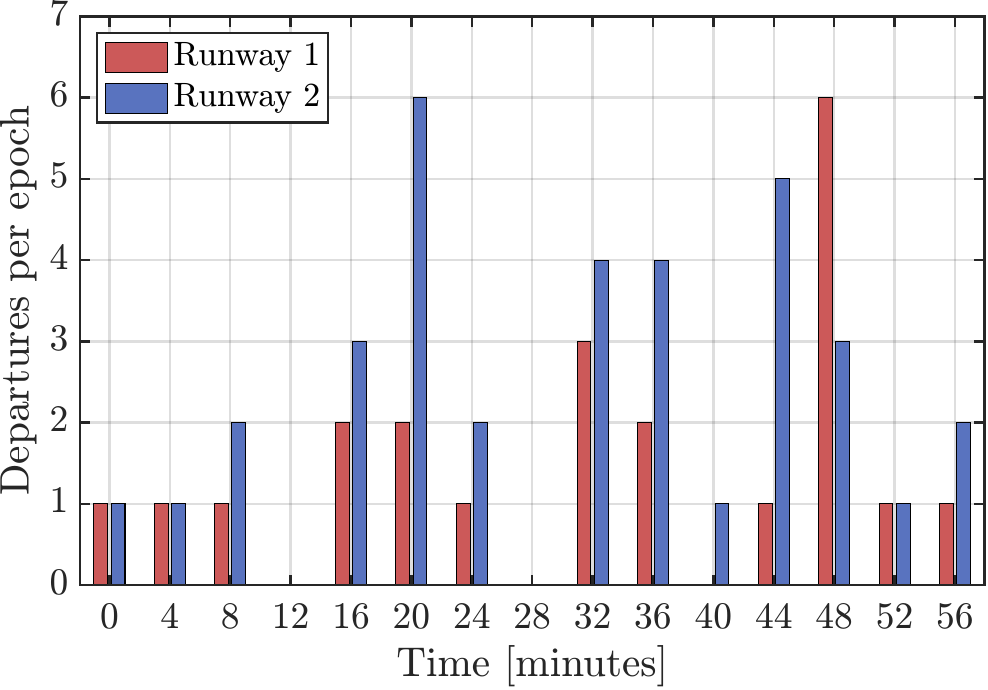}
    \caption{Departure rates per runway for each 4-minute epoch across a 60-minute horizon. Red and blue bars correspond to Runway~1 and Runway~2, respectively. The schedule is synthetically generated with an average departure rate of 1.05 flights per minute.}
    \label{fig:departure}
\end{figure}

The multi-epoch experiment evaluates rolling deployment over a one-hour horizon (16 epochs).
The system starts from empty queues, $q(0)=[0,0]$.
A single departure schedule is generated according to the following specification and kept fixed across experiments.
Aircraft arrivals follow a Poisson process with an arrival rate of 1.05 aircraft per minute, slightly exceeding the nominal system capacity of 1 aircraft per minute
(two runways, each serving 2 aircraft per 4-minute epoch).
Aircraft classes are generated with probabilities 0.4 (small), 0.3 (medium), and 0.3 (heavy), and each flight is randomly assigned to one of five airlines.
A departure schedule used in the simulation is shown in \Cref{fig:departure}.
Across experiments, only the cost noise level $\sigma$ is varied, while the underlying traffic instance remains unchanged.\footnote{If a correlated equilibrium-based solver fails to compute a feasible solution (e.g., no equilibrium found or computation time exceeds the epoch duration), the mechanism defaults to \texttt{FCFS} for that epoch.}

\begin{remark}[Implementation details]
    The experiments were implemented in Julia v1.11.3. We used the packages \texttt{ParametricMCP} \cite{parammcp} and \texttt{PATHSolver} \cite{PATH} with default convergence tolerances to compute correlated equilibria. All simulations were executed on a Linux computer equipped with an Intel Core i7-12700 CPU and 16 GB RAM. Each Monte Carlo experiment consisted of 100 independent trials with fixed random seeds for reproducibility.
\end{remark}

\subsection{Single-epoch experiment result}

\begin{figure}[hbt!]
    \centering
    \includegraphics[width=0.85\linewidth]{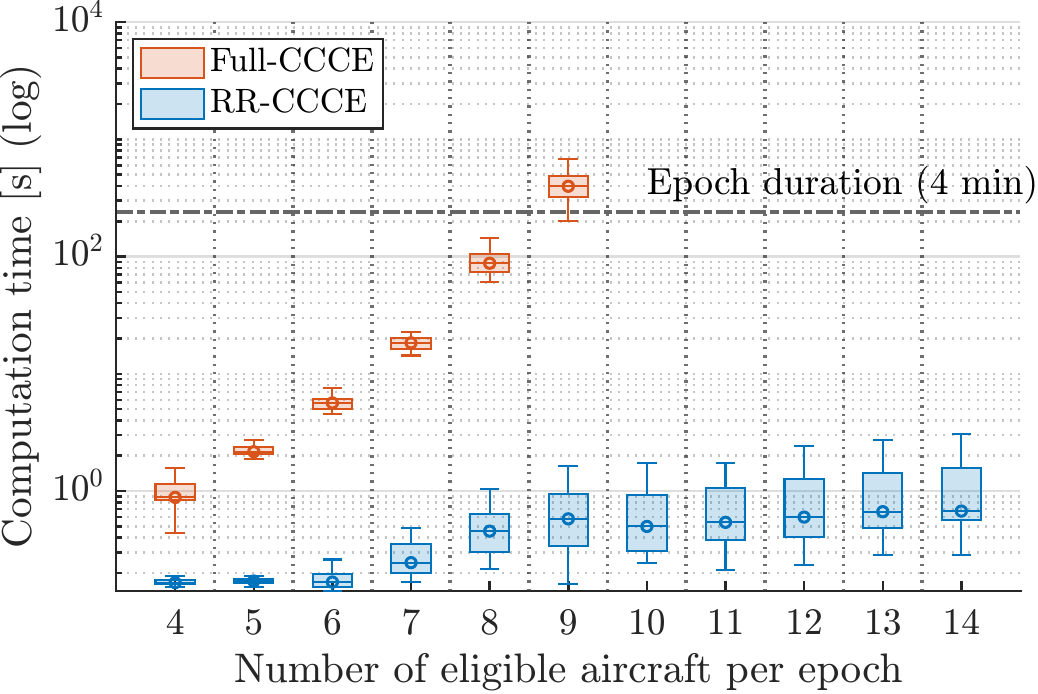}
    \caption{Scalability comparison between \texttt{Full-CCCE} and \texttt{RR-CCCE}.
    Wall-clock computation time (log scale) is shown as a function of the number of eligible aircraft per epoch.
    The horizontal dotted line indicates the epoch duration (4 minutes), representing the real-time constraint for online deployment. Median marked by $\circ$.}
    \label{fig:scalability}
\end{figure}

\begin{figure}[hbt!]
    \centering
    \includegraphics[width=0.85\linewidth]{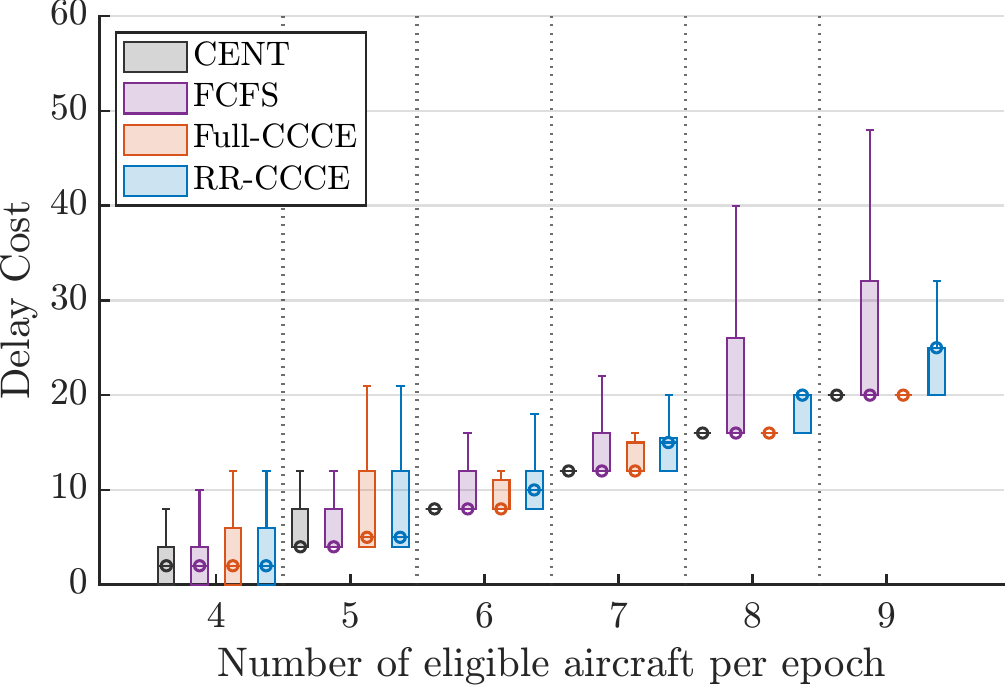}
    \caption{Comparison of realized delay cost across coordination mechanisms.
\texttt{CENT} serves as an idealized performance benchmark.
\texttt{Full-CCCE} and \texttt{RR-CCCE} outperform \texttt{FCFS} as traffic increases.
\texttt{RR-CCCE} exhibits a performance gap relative to \texttt{Full-CCCE} due to equilibrium-set approximation. Median marked by $\circ$.}
    \label{fig:delayCost}
\end{figure}

\subsubsection{Scalability benefit of CC-CE approximation}
To assess the scalability of \texttt{RR-CCCE}, we compare it with \texttt{Full-CCCE}.
\Cref{fig:scalability} shows wall-clock computation time as the number of eligible aircraft increases from 6 to 14.
For \texttt{Full-CCCE}, computation time grows rapidly due to the exponential expansion of the joint-action space, which scales on the order of $2^{|\fSet(t)|}$.
\texttt{RR-CCCE} remains well below the 4-minute epoch duration, while \texttt{Full-CCCE} exceeds this real-time threshold at 9 aircraft per epoch.

\subsubsection{Delay cost}

We compare the realized delay cost of the solution outcomes across solvers under $\sigma=0$.
\Cref{fig:delayCost} shows the cost as the number of eligible aircraft increases.

\texttt{CENT} serves as an idealized benchmark assuming full control authority and is included for reference.
At low traffic levels, medians are similar, though \texttt{Full-CCCE} and \texttt{RR-CCCE} show higher variance.
As traffic increases, the costs for \texttt{FCFS} rises sharply whereas uncertainty-aware methods yield lower costs. In the high-volume regime (8–9 aircraft per epoch), \texttt{Full-CCCE} achieves performance comparable to \texttt{CENT}.
\texttt{RR-CCCE} exhibits a modest performance degradation relative to \texttt{Full-CCCE}, due to the reduced-rank approximation of the equilibrium set.

\subsubsection{Performance under uncertainty}

\begin{figure}[t]
  \centering
    \begin{subfigure}{0.425\textwidth} \label{fig:eff_sigma}
        \includegraphics[width=\linewidth]{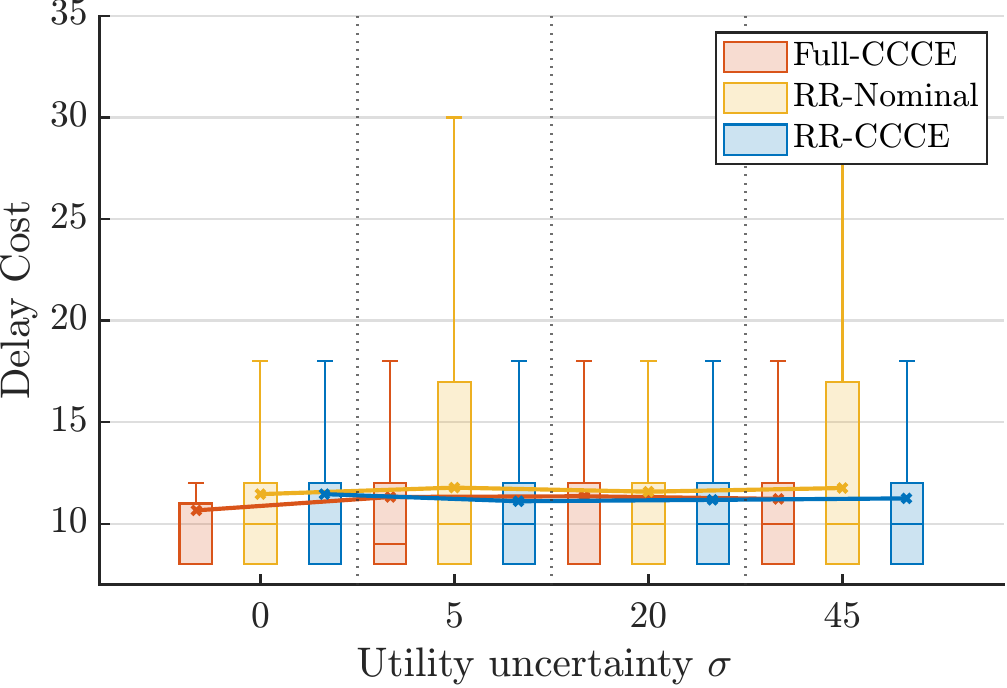}
        \caption{}
    \end{subfigure}
    \begin{subfigure}{0.425\textwidth} \label{fig:devfreq_sigma}
        \includegraphics[width=\linewidth]{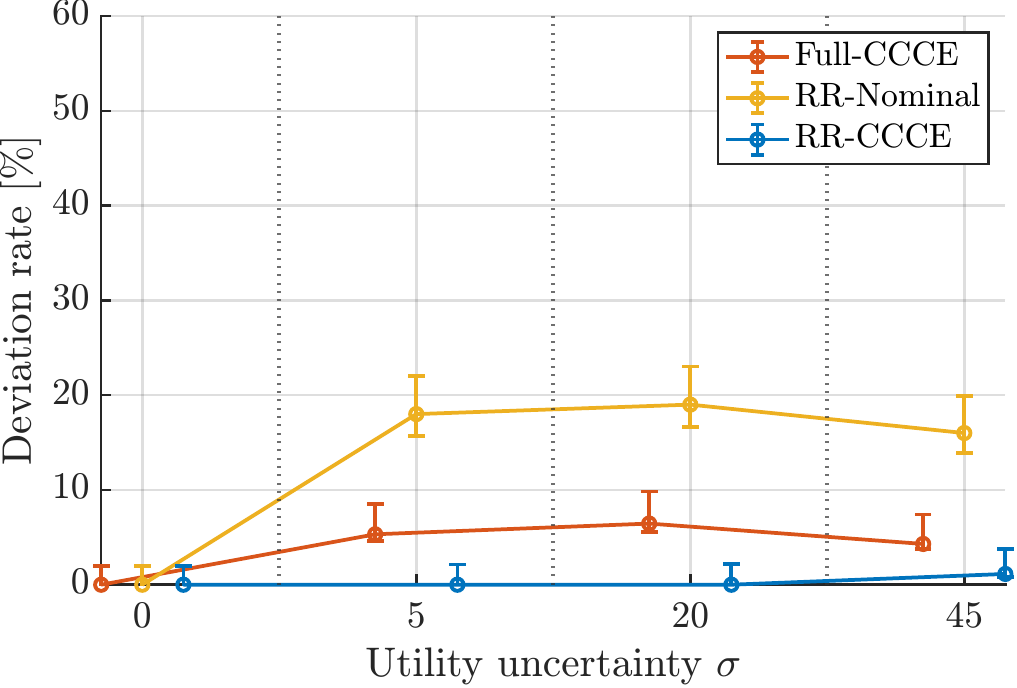}
        \caption{}
    \end{subfigure}
  \caption{Performance under cost uncertainty ($\alpha=90\%$, six airlines).
(a) Delay cost distribution as a function of cost noise level $\sigma$.
(b) Deviation rate versus $\sigma$.
Horizontal lines within boxplots indicate mean values.}

  \label{fig:sigma_2stack}
\end{figure}

\begin{figure}[t]
  \centering
    \begin{subfigure}{0.425\textwidth} \label{fig:eff_alpha}
        \includegraphics[width=\linewidth]{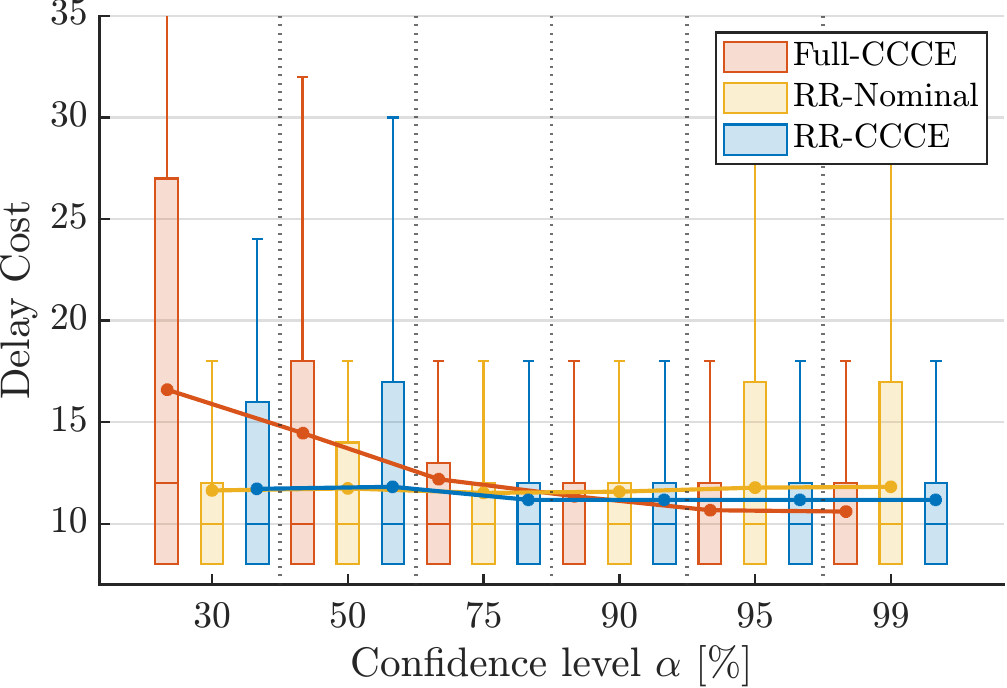}
        \caption{}
    \end{subfigure}
    \begin{subfigure}{0.425\textwidth} \label{fig:devfreq_alpha}
        \includegraphics[width=\linewidth]{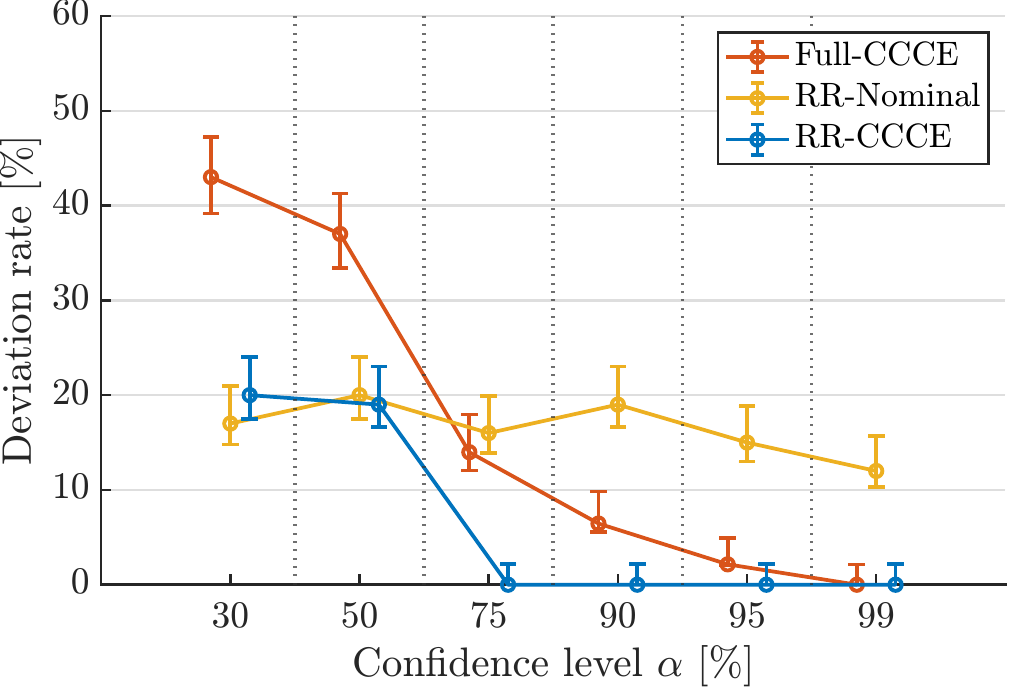}
        \caption{}
    \end{subfigure}
  \caption{Performance under varying confidence levels ($\sigma=20$, six airlines).
(a) Delay cost distribution as a function of confidence level $\alpha$.
(b) Deviation rate versus $\alpha$.
Horizontal lines within boxplots indicate mean values.}
  \label{fig:alpha_2stack}
\end{figure}

We evaluate uncertainty robustness by comparing \texttt{Full-CCCE}, \texttt{RR-Nominal}, and \texttt{RR-CCCE}.
The number of airlines is fixed to six.
When varying the cost noise level $\sigma$, the confidence level is fixed at $\alpha=90\%$.
When varying $\alpha$, the noise level is fixed at $\sigma=20$.

\paragraph{Effect of cost uncertainty ($\sigma$)}

As shown in \Cref{fig:sigma_2stack}, the mean delay costs remain comparable across all methods over the range of uncertainty levels.
While most algorithms exhibit relatively stable cost distributions,
\texttt{RR-Nominal} shows a mild increase in variance in the high-uncertainty regime.
Overall, the delay performance remains largely consistent across methods as $\sigma$ increases.

Regarding deviation frequency, uncertainty-aware methods demonstrate improved robustness.
Notably, under $\alpha=90\%$, \texttt{Full-CCCE} exhibits deviation rates roughly aligned with the confidence level, whereas \texttt{RR-CCCE} maintains near-zero deviation frequency across all tested $\sigma$ values.

\paragraph{Effect of confidence level ($\alpha$)}

As the confidence level $\alpha$ increases, deviation rates for \texttt{Full-CCCE} start from relatively high levels and gradually decline.
\texttt{RR-CCCE} maintains the lowest deviation rate overall and approaches zero deviation for $\alpha \ge 75\%$.
In contrast, \texttt{RR-Nominal} remains with deviation rates consistently in the 10--20\% range across $\alpha$.
These trends are illustrated in \Cref{fig:alpha_2stack}.

In terms of delay cost, \texttt{Full-CCCE} shows a decreasing trend as $\alpha$ increases. The mean costs for the others remain relatively stable across $\alpha$, with the smallest variance and mean value observed around $\alpha \in [0.90, 0.99]$.

\begin{figure}[hbt!]
    \centering
    \includegraphics[width=0.89\linewidth]{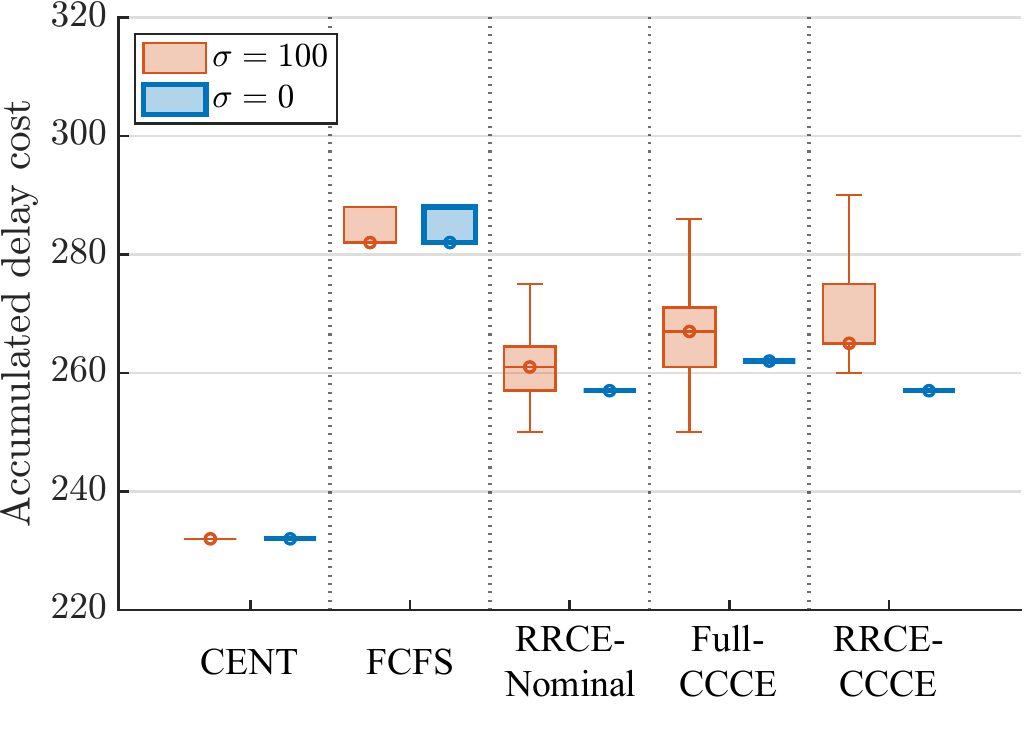}
    \caption{Accumulated delay cost over a 16-epochs rolling simulation.
    Two uncertainty settings are compared: $\sigma=0$ and $\sigma=100$.
    Each boxplot summarizes repeated runs with identical traffic realization and varying cost noise. Median marked by $\circ$.
    }
    \label{fig:multiepoch}
\end{figure}

\subsection{Multi-epoch experiment result}

We evaluate the rolling deployment of each coordination mechanism over a one-hour horizon (16 epochs).
The accumulated cost over the horizon reflects repeated application of the single-epoch policy rather than a long-horizon optimal solution.
This experiment evaluates cumulative effects and operational scalability.

Across both high- and low-$\sigma$ scenarios, correlated equilibrium-based methods achieve approximately 6.0--7.4\% reduction in accumulated delay compared to \texttt{FCFS} as shown in \Cref{fig:multiepoch}. When $\sigma=0$, this improvement increases to approximately 7.1--8.9\%.
Under large uncertainty ($\sigma=100$), \texttt{RR-Nominal} sometimes achieves lower accumulated cost than uncertainty-aware variants. 
This reflects a trade-off between diverse profile selection and robustness: nominal coordination can exploit high-performing equilibria but tolerates higher deviation risk.

\section{Discussion}

\begin{figure}[t!]
  \centering
    \begin{subfigure}{0.425\textwidth} 
        \includegraphics[width=\linewidth]{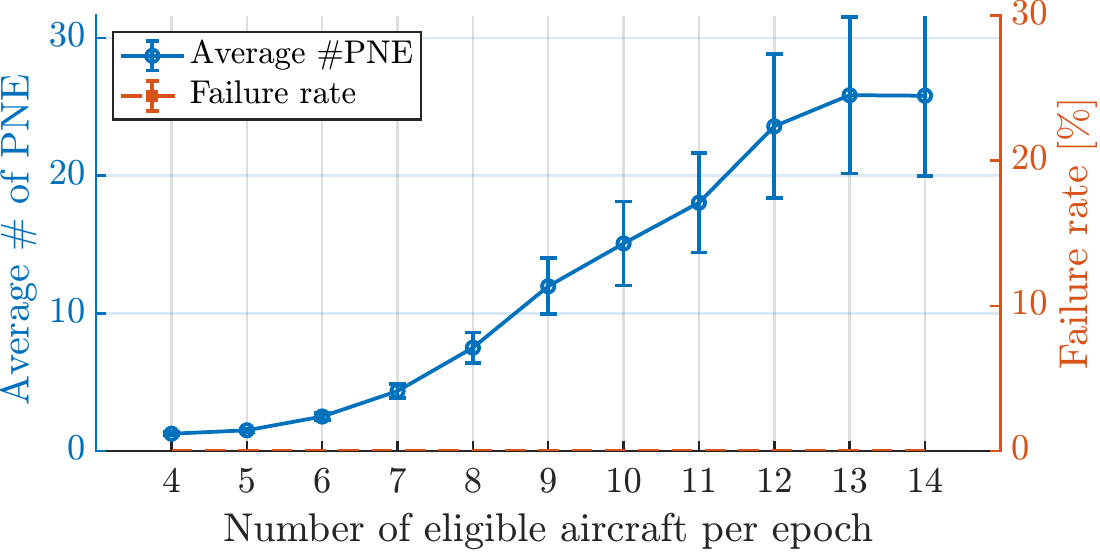}
        \caption{}
    \end{subfigure}
    \begin{subfigure}{0.425\textwidth} 
        \includegraphics[width=\linewidth]{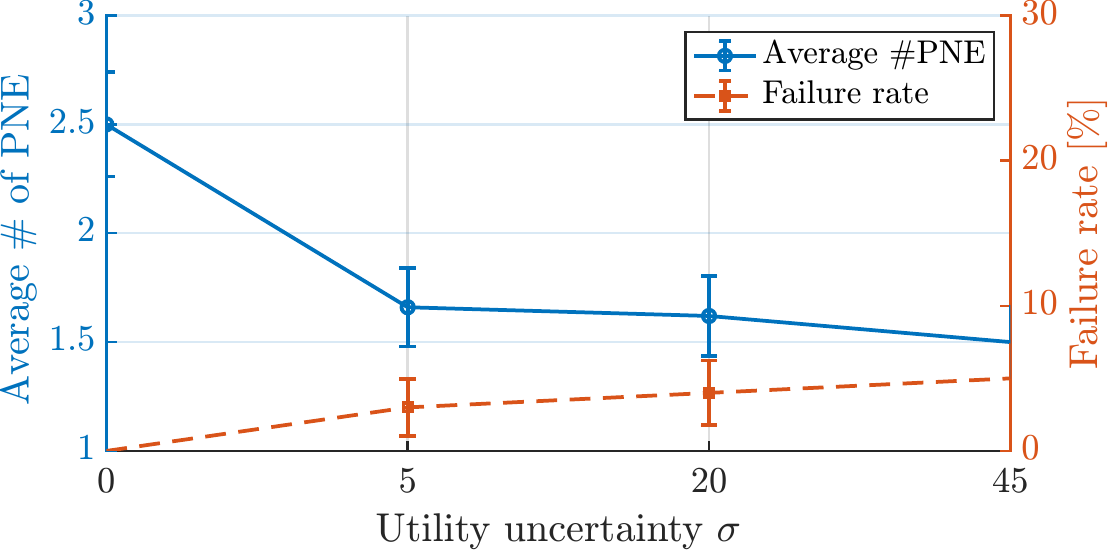}
        \caption{}
    \end{subfigure}
  \caption{Equilibrium-space characteristics.
    (a) Average number of chance-constrained pure Nash equilibria (CC-PNE) and solver failure rate as the number of eligible aircraft increases ($\alpha=90\%$, $\sigma=20$).
    (b) Average number of CC-PNE and failure rate as cost uncertainty $\sigma$ increases ($|\fSet(t)|=8$, $\alpha=90\%$).
    Increasing problem size expands the equilibrium set, whereas higher uncertainty reduces discoverable equilibria and increases failure probability.}
  \label{fig:pne_error_2stack}
\end{figure}

\subsection{Coordination benefit over first-come-first-served (FCFS)}

Across both low- and high-uncertainty settings, correlated equilibrium-based coordination consistently reduces accumulated delay compared to \texttt{FCFS}, which reflects current CVQ first-come-first-served operations.
This confirms that providing incentive-compatible aircraft-level recommendations improves system-level performance without requiring centralized control authority.
The improvement is more pronounced in low-uncertainty environments, where deviation is limited and recommended profiles are more reliably implemented.

\subsection{Confidence level, deviation, and feasible set}

A key observation concerns the interaction between confidence level $\alpha$, deviation rate, and cost performance under uncertainty.
As expected, increasing $\alpha$ reduces the empirical deviation rate.
This reflects the widening of incentive margins in the chance-constrained equilibrium conditions.
At the same time, increasing $\alpha$ contracts the feasible CC-CE region, and in some instances reduces the number of available equilibrium profiles or increases solver failure rates.
A similar contraction effect is observed when cost uncertainty $\sigma$ increases as illustrated in \Cref{fig:pne_error_2stack}.

While the relationship between $\alpha$ and both deviation rate and the size of the feasible equilibrium region is structurally clear, its effect on realized cost is more nuanced.
Reducing deviation increases the likelihood that the recommended profile is implemented,
but shrinking the feasible equilibrium region may eliminate high-performing profiles.
This trade-off is visible in the experiments.

In the single-epoch setting, \texttt{RR-CCCE} and \texttt{Full-CCCE} improve median cost as $\alpha$ increases from 30\% to 75\% and the deviation rate drops (\Cref{fig:alpha_2stack}).
However, the performance of \texttt{Full-CCCE} slightly exceeds that of \texttt{RR-CCCE} in the $\alpha \in [0.95, 0.99]$ regime, and even \texttt{Full-CCCE} emits a nonzero deviation rate.
A similar pattern is observed in the multi-epoch setting, where \texttt{RR-Nominal} sometimes achieves lower accumulated cost under high noise,
at the expense of a higher deviation frequency. This does not contradict the robustness objective; rather, it highlights the risk–performance trade-off inherent in chance-constrained coordination.

\subsection{Algorithmic trade-offs}

Differences between full and reduced formulations further highlight this structural trade-off.
\Cref{fig:alpha_2stack} shows that \texttt{RR-CCCE} exhibits substantially lower deviation rates than \texttt{Full-CCCE} at identical confidence levels.
This can be interpreted as a consequence of the reduced space of equilibria:
the reduced-rank approximation limits the available profiles but implicitly increases robustness.
Overall, the results suggest the existence of a confidence-level ``sweet spot'', balancing profile diversity and implementation reliability.

\section{Conclusion}

We proposed a chance-constrained correlated equilibrium framework for collaborative virtual queue coordination under airline cost uncertainty.
The formulation provides probabilistic incentive guarantees and allows explicit adjustment of confidence levels.
A reduced-rank approach was developed to enable scalable computation.

Numerical experiments demonstrate scalability up to traffic levels corresponding to 210 aircraft eligible for pushback per hour—consistent with operations at the world's busiest airports—and show up to an 8--9\% reduction in accumulated delay compared to the current first-come-first-served mechanism.
The results further reveal a trade-off among confidence level, deviation frequency, and cost efficiency.

Future work includes characterizing conditions under which pure Nash equilibria cease to exist under high uncertainty and systematically identifying optimal confidence-level operating points that balance robustness and performance.

\bibliographystyle{IEEEtran} 
\bibliography{reference}

\vspace{12pt}

\end{document}